\documentclass[12pt,titlepage]{article}
\usepackage{latexsym,graphics,epsfig,amsmath,amssymb,tabularx,booktabs}

\usepackage[authoryear]{natbib}

\usepackage{enumerate}
\usepackage{amsthm}
\usepackage{url}
\usepackage{fancyhdr,titling}

\usepackage[T1]{fontenc}
\theoremstyle{definition}

\theoremstyle{remark}

\theoremstyle{plain}
\newtheorem{theorem}{Theorem}[section]
\newtheorem{lemma}[theorem]{Lemma}
\newtheorem{proposition}[theorem]{Proposition}
\newtheorem{corollary}[theorem]{Corollary}

\def\keywords#1{{\vskip4pt
\noindent
\hbox to59.5pt{KEY\enspace WORDS:\quad\hss}\vtop{\advance \hsize by -59.5pt
\leftskip=28pt \rightskip=0pt
\noindent\ignorespaces#1\vskip8pt}}}

\makeatletter
\let\runauthor\@author
\let\runtitle 

\title{Iterative Scaling in Curved Exponential Families}

\author{Anna Klimova \\ 
Institute of Science and Technology (IST) Austria
\and Tam\'{a}s Rudas \\
E\"{o}tv\"{o}s Lor\'{a}nd University, Budapest, Hungary}
\date{}
\begin{document}

\begin{titlepage}
\thispagestyle{empty}
{ \hfill Iterative scaling in curved families}

\begin{center}
~\\[5cm]
\textsc{\large{Iterative Scaling in Curved Exponential Families}}\\[1.5cm]


\large{Anna {Klimova} ~\\

Institute of Science and Technology (IST) Austria, ~\\
Klosterneuburg, Austria

\vspace{1cm}

Tam\'{a}s {Rudas} ~\\
E\"{o}tv\"{o}s Lor\'{a}nd University, ~\\
Budapest, Hungary}

\date{}

\end{center}

\end{titlepage}

\newpage


\begin{abstract}
The paper describes a generalized iterative proportional fitting procedure which can be used for maximum likelihood estimation in a special class of the general log-linear model. The models in this class, called relational,  apply to multivariate discrete sample spaces which do not necessarily have a Cartesian product structure and may not contain an overall effect.  When applied to the cell probabilities, the models without the overall effect are curved exponential families and the values of the sufficient statistics are reproduced by the MLE only up to a constant of proportionality. The paper shows that Iterative Proportional Fitting, Generalized Iterative Scaling and Improved Iterative Scaling, fail to work for such models. The algorithm proposed here is based on iterated  Bregman projections. As a by-product, estimates of the multiplicative parameters are also obtained.


\end{abstract}


\baselineskip=18pt

\section*{Introduction}

This paper deals with variants of the general log-linear model
\begin{equation}\label{genll}
\log \boldsymbol{\delta} = \textbf{A}' \boldsymbol{\beta},
\end{equation}  
where $\boldsymbol{\delta}$ denotes a vector of probabilities (appropriate in the case of  multinomial sampling) or a vector of intensities (appropriate in the case of Poisson sampling) and the model matrix $\mathbf A$ has non-negative integer entries.

When the sample space is the Cartesian product of the ranges of categorical variables and the rows of $\textbf{A}$ are indicators  of cylinder sets of the sample space, (\ref{genll}) specifies a conventional log-linear model \citep*[cf.][]{BFH}. When the sample space does not have a  Cartesian product structure and $\textbf{A}$ is the indicator matrix of arbitrary subsets of the cells in the sample space,  (\ref{genll}) defines a relational model \citep*{KRD11}.

Models that associate parameters with subsets of cells appear, for instance, in  areas of machine learning which deal with feature selection. Features are characteristics of objects, and a subset of cells comprises the objects that possess a particular feature. The goal of the analysis is to choose features so that a Markov field  based on them approximates the observed distribution well, see, e.g., \cite*{LaffertyBregman99,{LaffertyMccallumPereira}, Malouf, Huang2010}. Maximum entropy models used in machine learning, see, e.g., \cite{LaffertyMccallumPereira}, are special cases of relational models. Feature selection techniques are used, among others,  in text processing \citep*[cf.][]{Mccallum2000}, in computer tomography \citep[cf.][]{OSullivan}, and in the analysis of social mobility \citep[cf.][]{KRbm}.

Many problems in feature selection allow for the existence of unaffected cases, i.e., objects who do not possess any of the characteristics of interest. However, in some contexts of feature selection, such an assumption may not be feasible. For example, in market basket analysis \citep*[cf.][]{Brin1997,Wu2003}, where records of purchases are analyzed to reveal patterns of associations among the different goods bought, each purchase consists of one item, at least.  Similarly, a registry of congenital abnormalities \citep[cf.][]{KallenRegistry, Anomalies2001Boyd, Anomalies2007} lists only affected newborns.  While in the case of birth defects, unaffected newborns exist and their total number may be known, in market basket analysis there is no purchase with nothing bought. In such problems, testing hypotheses of association between features, e.g., independence, cannot be performed using conventional log-linear models, but may be done within the relational model framework \citep{KRD11}.

In Section \ref{RMchapter}, the formal definition and the main properties of relational models are reviewed. The characteristics of model (\ref{genll}) are affected by the presence or absence of  the vector of $1$'s, denoted in the sequel as $\boldsymbol{1}$,  in the row space $R(\textbf{A})$ of the model matrix.  If $\boldsymbol{1} \in R(\textbf{A})$, there exists a parameterization of the model in which one of the parameters appears in every cell; such a model is said to have the overall effect. Relational models for probabilities with the overall effect and relational models for intensities  are regular exponential families, and standard results about the MLE in such families apply. Relational models for probabilities without the overall effect are curved exponential families, and some properties of the MLE are fundamentally different.   

In Section \ref{GISIISsection}, three  iterative scaling algorithms used for models of type (\ref{genll}) are reviewed and their applicability for determining the MLE under relational models is investigated. The iterative proportional fitting (IPF) procedure   \citep{DemingStephan} is employed for conventional log-linear models \citep[cf.][]{BFH} and can be easily modified to suit relational models with the overall effect. Generalized Iterative Scaling (GIS), proposed by \cite{DarrochRatcliff}, is used in feature selection and works under an assumption which is sometimes called ``a constant sum of features''. It is shown, that this assumption is equivalent to the presence of the overall effect in the model, and thus GIS can only be applied to relational models with the overall effect. Improved Iterative Scaling (IIS) was proposed by \cite*{DDL1997} as a generalization of GIS which does not rely on the assumption of constant sum of features. One might expect that IIS could be used for relational models without the overall effect, but it is shown that this is not the case.         

The main contribution of the paper, described in Section \ref{MLEipf}, is an iterative proportional fitting procedure that generalizes the traditional IPF and can be used for all relational models, with or without the overall effect. The algorithm constructs a sequence of projections on convex sets, and, while for the traditional IPF procedure the projections minimize the Kullback-Leibler distance, the projections for the new algorithm minimize the Bregman distance \citep{Bregman}. Iterative proportional fitting is typically used for estimating the cell parameters and, to the best knowledge of the authors, it is not used to estimate the model parameters. It is also shown here, that the estimates of the model parameters can be found as a by-product of the IPF procedure. A numerical variant of the generalized IPF is also described, and it is proven that its output approximates the true limit with any specified precision.


 
\section{Preliminaries} \label{RMchapter} 

Let $Y_1, \dots, Y_K$ be  random variables taking values in finite sets $\mathcal{Y}_1, \dots, \mathcal{Y}_K$, respectively. A combination of values $(y_1, y_2,\dots,y_K) \in \mathcal{Y}_1 \times \dots \times \mathcal{Y}_K$ is called a cell, and a 
 sequence $\mathcal{I}$ of cells is a table. The population distribution is parameterized by $\boldsymbol \delta =\{\delta(i),\,\,\mbox{for } i \in \mathcal{I}\}$. Depending on whether the population is characterized by probabilities  or by intensities, $\delta(i) \equiv p(i) \in (0,1)$, with $\sum_{i \in \mathcal{I}} p(i) = 1$, or $\delta(i) \equiv \lambda(i) > 0$, for all $i \in \mathcal{I}$. The distinction between data collected through multinomial or Poisson sampling procedures is necessary, because maximum likelihood estimates under the two sampling schemes may have fundamentally different properties under the models considered here.  

Let $\mathcal{P}$ be the set of positive distributions, parameterized by $\boldsymbol \delta$, on $\mathcal{I}$. {A  relational model} is generated by a class $\mathbf{S} = \{S_1, \dots, S_J\}$ of non-empty subsets of the table $\mathcal{I}$. A distribution ${P}_{\boldsymbol \delta} \in \mathcal{P}$ is in the model if and only if
\begin{equation} \label{PMmatr}
\mbox{log } \boldsymbol \delta = \mathbf{A}'\boldsymbol \beta, \, \mbox{for some } \, \boldsymbol \beta \in \mathbb{R}^J.
\end{equation}
Here, the components of $\boldsymbol \beta$ are the log-linear parameters of the model and the rows of the model matrix $\mathbf{A}$ are indicators of the generating subsets.  

A relational model is also an exponential (multiplicative) family 
\begin{equation}\label{dualF2}
\mathcal{M}(\mathbf{A}) = \{ P_{\boldsymbol \delta} \in \mathcal{P}: \,\, \boldsymbol \delta = {\boldsymbol \theta}^{\mathbf{A}'}, \mbox{for some } \boldsymbol \theta \in \mathbb{R}^J_{>0}\},
\end{equation}
where, for every $i \in \mathcal{I}$, $\,\delta(i) = \prod_{j=1}^J \theta_j^{a_{ji}}$, and the components of $\boldsymbol \theta$ are the multiplicative parameters.

The overall effect, present in every cell, plays the role of the normalizing constant, say $\beta_0$, which is often included in exponential family models:
$$\boldsymbol \delta = \frac{1}{\beta_0}\mbox{exp }\{\mathbf{A}'\boldsymbol \beta\}.$$
Such models can be re-written in the form (\ref{PMmatr}) by adding a row of $1$'s to $\mathbf{A}$. Therefore, relational models do not assume that there is no normalizing constant, they only allow for this possibility. 

A dual representation of a relational model is obtained using  an integer kernel basis matrix $\mathbf{D}$ \citep{KRD11}, whose rows are a basis of $Ker(\mathbf{A})$. Thus, (\ref{dualF2}) can be re-written as 
\begin{equation}\label{dualF2d}
\mathcal{M}(\mathbf{A}) = \{ P_{\boldsymbol \delta} \in \mathcal{P}: \,\, \mathbf{D} \mbox{log }\boldsymbol \delta = \boldsymbol 0\}.
\end{equation}
The dual representation says that the generalized odds ratios,
\begin{equation} \label{defORd}
\mathcal{OR}_{\boldsymbol d} = \boldsymbol \delta ^{\boldsymbol {d^+}}/\boldsymbol \delta ^{\boldsymbol {d^-}},
\end{equation}
associated with the rows of $\mathbf{D}$, are all equal to 1, 
where ${\boldsymbol {d^+}}$ and ${\boldsymbol {d^-}}$ denote, respectively, the positive and negative parts of row $\boldsymbol d$. If the degrees of the numerator and of the denominator in (\ref{defORd}) are equal, the odds ratio $\mathcal{OR}_{\boldsymbol d}$ is called homogeneous and, otherwise it is called non-homogeneous. For a relational model with the overall effect, all generalized odds ratios are homogeneous, and it can be shown that any model without the overall effect has a dual representation with exactly one non-homogeneous odds ratio. 

Some of the properties of maximum likelihood estimates under relational models are also affected by presence or absence of the overall effect. Let $\mathbf{Y}$ be a random variable that has a distribution parameterized by $\boldsymbol \delta$ and let $\boldsymbol y$ be a realization of $\mathbf{Y}$. Assume that the MLE $\hat{\boldsymbol \delta}$, under the relational model generated by $\mathbf{A}$, exists. The properties of the MLE in $\mathcal{M}(\mathbf{A})$ are best described by using the linear family
$$
\mathcal{L}_{\boldsymbol \delta}(\mathbf{A}, \boldsymbol q, \gamma) = \{ P_{\boldsymbol \delta} \in \mathcal{P}: \,\, \mathbf{A} \boldsymbol \delta = \gamma \mathbf{A} {\boldsymbol q}\},
$$
where
\begin{equation}\label{b}
\boldsymbol q = \left\{\begin{array}{ll} \boldsymbol y, & \mbox{if } \boldsymbol \delta \equiv \boldsymbol \lambda, \\
\boldsymbol y/ (\boldsymbol 1\boldsymbol y), & \mbox{if } \boldsymbol \delta \equiv \boldsymbol p.
 \end{array}\right.
\end{equation}
As was shown by \cite{KRD11}, when the relational model is a regular exponential family, the distribution parameterized by $\hat{\boldsymbol \delta}$ is the unique point in  
\begin{equation}\label{linF1}
\mathcal{M}(\boldsymbol A) \cap \mathcal{L}_{\boldsymbol \delta}(\mathbf{A}, \boldsymbol q, 1),
\end{equation}
and, when the relational model is a curved exponential family,  the distribution parameterized by $\hat{\boldsymbol \delta}$ is the unique common point in
\begin{equation}\label{linF1c}
\mathcal{M}(\boldsymbol A) \cap \mathcal{L}_{\boldsymbol \delta}(\mathbf{A}, \boldsymbol q, \gamma),
\end{equation}
for the unique $\gamma > 0$ for which $\boldsymbol 1\hat{\boldsymbol \delta} = 1$. The components of $\mathbf{A}\boldsymbol \delta$ are called the subset sums of $P_{\boldsymbol \delta}$. The coefficient of proportionality $\gamma$ is called the adjustment factor.  The properties of maximum likelihood estimates under relational models are summarized in Table \ref{ModelTypes}.

\begin{center}  *** Table \ref{ModelTypes} here.  *** \end{center}

In general, there are no closed form expressions for the MLE. Maximum likelihood estimates for the cell parameters can be computed using the Newton-Raphson algorithm or    algorithms  for convex optimization \citep[cf.][]{BertsekasNLP, AitchSilvey60, EvansForcina11}. This paper focuses on the iterative scaling approach to computing the MLE, as it is often used in feature selection procedures \citep[cf.][]{Huang2010}.

\section{Why IPF, GIS, IIS cannot be used for relational models?} \label{GISIISsection}

The traditional iterative proportional fitting (IPF) procedure is used for computing maximum likelihood estimates of the cell frequencies or probabilities under log-linear models \citep[cf.][]{BFH}. The IPF algorithm starts with a contingency table of the same structure as the sample space, with all cell frequencies equal to one: ${\boldsymbol \delta}^{(0)}=\boldsymbol 1$. The cell frequencies are then adjusted until the marginal sums $A_{j}\boldsymbol{\delta}^{(d)}$, where $d+1 \equiv j \mbox{ mod } J$, become equal or close enough to the observed values: 
\begin{equation}\label{classicIPF} 
\delta^{(d+1)}(i) = \delta^{(d)}(i) \left(\frac{A_{j}\boldsymbol{q}}{A_{j}\boldsymbol{\delta}^{(d)}}\right)^{a_{ji}}, \,\,  \mbox{for  all } i \in \mathcal{I}.
\end{equation}
The model structure, expressed in terms of the odds ratios, is preserved during iterations.
Given that the MLEs of the cell parameters exist, the sequence $\boldsymbol \delta^{(d)}$ converges, as $d \to \infty$, to the maximum likelihood estimate $\hat{\boldsymbol \delta}$. 
The proof of convergence of IPF relies on a particular order of the subsets. In this order, the cylinder sets associated with the $k$-th marginal are listed in one block indexed from $j=[k]_s$ to $j=[k]_e$. Then the distribution parameterized by $\delta^{(tJ+[k]_e)}$ is the I-projection of the distribution parameterized by $\delta^{(tJ+[k-1]_e)}$ in $\cap_{j=[k]_s}^{[k]_e}\mathcal{L}_{\boldsymbol \delta}(A_{j}, \boldsymbol q, 1)$, for every non-negative integer $t$, see \cite{Csiszar}. 

A straightforward modification of IPF, which updates, instead of marginal totals, the subset sums, may be considered for computing the MLEs of the cell parameters under relational models.  This will  work when the overall effect is present in the model, however, the proof needs modification, because in the actual parameterization, there may be no sets $S_j$ that form a partitioning of the sample space, like $S_{[k]_s}, \ldots , S_{[k]_e}$ do, and, therefore, the algorithm may not produce I-projections. When the overall effect is not present,  (\ref{classicIPF}) works for intensities, but as the total is not reproduced by the MLE, a new proof is needed, see Section \ref{MLEipf}.


The generalized iterative scaling (GIS) algorithm was proposed by \cite{DarrochRatcliff} for maximum likelihood estimation in discrete exponential families of the form $\mbox{log }\boldsymbol p = \mathbf{A}'\boldsymbol \beta$, where  $\mathbf{A}$ is a non-negative real matrix whose rows add to the vector of $1$'s:
\begin{equation}\label{star}
A_1 + \dots + A_J = (\sum_{j = 1}^J a_{j1},\dots, \sum_{j = 1}^J a_{j|\mathcal{I}|}) = \boldsymbol 1.
\end{equation}
The GIS procedure organizes updating cycles not according to subsets (or marginals), rather according to cells. It performs all adjustments that apply to a cell in one step, and recalculates the multipliers to be used in the next cycle of adjustment only after all cells were updated: 
\begin{equation} \label{DRIter} 
p^{(n+1)} ({i}) = p^{(n)}( i) \prod_{j =1}^J \left( \frac{ A_j\boldsymbol q}{ A_j\boldsymbol p^{(n)}}\right)^{a_{ji}}, \hspace{2mm}  \mbox{ for all } i \in \mathcal{I}.
\end{equation}
As $n \to \infty$, $\boldsymbol p^{(n)}$ converges to the maximum likelihood estimate $\hat{\boldsymbol p}$, see Theorem 1 in \cite{DarrochRatcliff}.  The condition (\ref{star}) is used in the proof of convergence, but it can be shown that GIS also applies to relational models with a model matrix whose rows sum to a constant vector:
\begin{equation}\label{starr}
A_1 + \dots + A_J = c\boldsymbol 1.
\end{equation}
In fact, (\ref{starr}) and the presence of the overall effect are equivalent:

\begin{proposition}\label{featuresVSoverEff}

For a 0-1 matrix $\boldsymbol A$, the following two conditions are equivalent:
\begin{enumerate}[(i)]
\item \label{i1r} $\boldsymbol{1} \in R(\mathbf{A})$.
\item \label{ii1r} There exists a non-negative matrix $\tilde{\mathbf{A}}$, with $R(\tilde{\mathbf{A}}) = R(\mathbf{A})$, whose rows add to a vector $ c\boldsymbol 1$, for some $c \in \mathbb{Z}_{+}$.
\end{enumerate}
\end{proposition}

\begin{proof}

Suppose (\ref{i1r}) holds. Without loss of generality, assume that $A_1 = \boldsymbol 1$. The matrix $\mathbf A^*$ consisting of  rows
\begin{align} \label{NewBasis}
&A_1^* = A_1 - A_2, \quad A_2^* = A_2 - A_3,  \quad \dots,  \quad A_{J-1}^* = A_{J-1} - A_{J},  \quad A_J^* = A_J,
\end{align}
has the same row space as $\mathbf A$. The rows of $\mathbf A^*$ sum to $\boldsymbol 1$, but its entries 
may not be all positive. Let $a_j^* = \underset{i\in\mathcal{I}}{\operatorname {min}}(a_{ji}^*)$, for $j = 1, \dots, J$. 

The matrix $\tilde{\mathbf{A}}$ with rows $\tilde{A}_j = A_j^* - a_j^* \boldsymbol 1
$ is nonnegative, has the same row space as $\mathbf A$, and the rows of $\tilde{\mathbf{A}}$ sum to  
$c\boldsymbol 1$, where $c =1 -a_1^* - a_2^* - \ldots - a_J^* \in \mathbb{Z}_{+}$.

Suppose (\ref{ii1r}) holds. Then $\boldsymbol 1$ is a linear combination of some vectors in $R(\mathbf{A})$ and thus $\boldsymbol 1 \in R(\mathbf{A})$. The proof is now complete.
\end{proof}

In feature selection procedures, before GIS is applied, a model matrix with a non-constant sum of rows is often converted to a matrix with a constant row sum by adding a "slack feature", which balances other rows, to the model \cite[cf.][]{LaffertyMccallumPereira, Mccallum2000}. By Proposition \ref{featuresVSoverEff}, a model with a slack feature always has the overall effect. Adding a slack feature to a model without the overall effect  changes the model. 

Another algorithm, Improved Iterative Scaling (IIS),  proposed by \cite{DDL1997}, claims that it ``is an improvement of the Generalized Iterative Scaling algorithm of Darroch and Ratcliff in that it does not require that the features sum to a constant''. The IIS algorithm cyclically updates the cells of the table, and each iteration consists of two steps - updating the multipliers and updating the cell probabilities:
\begin{enumerate}
\item For each $j = 1, \dots, J$, solve
\begin{equation} \label{LaffIterSystem}
\sum_{i \in \mathcal{I}}a_{ji}p^{(n)}(i) (\zeta_j^{(n)})^{ a(i)} = A_j\boldsymbol q, \mbox{ where } a(i) = \sum_{j = 1}^J a_{ji},
\end{equation}
to compute $\boldsymbol \zeta^{(n)} = (\zeta^{(n)}_1, \dots,\zeta^{(n)}_J)'$.
\item Compute  
\begin{equation} \label{LaffIter}
p^{(n+1)} ({i}) = p^{(n)}( i) \prod_{j =1}^J (\zeta_j^{(n)})^{a_{ji}},
\end{equation}
and set $n = n+1$.
\end{enumerate}
In the case when all of $a(i) = \sum_{j = 1}^J a_{ji}$ are equal to 1, (\ref{LaffIter}) reduces to the updating step (\ref{DRIter}) of the GIS procedure. 

The literature on the IIS algorithm seems implicit about whether or not the model matrix has a row of  $1$'s. For example,  \cite{DDL1997} and \cite{Bancarz2002} mention that a normalizing constant is included in the model, however \cite{LaffertyBregman99} do not explicitly say it, but prove that the sequence $\boldsymbol p^{(n)}$ converges to the MLE. In fact, if the model has the overall effect, it does not matter whether (\ref{star}) or (\ref{starr}) holds in the actual parameterization. But if the model does not contain the overall effect, IIS may converge to a vector of probabilities that does not sum to 1 and thus can not be the MLE. Normalization of the limit vector would not help: a relational model without the overall effect is not scale invariant \citep{KRD11}, that is, a non-normalized vector of probabilities, that has the multiplicative structure prescribed by the model, looses this structure after being normalized. For example, the model of independence between three features  $A$, $B$, or $C$  is defined by  \cite{AitchSilvey60} as 
\begin{equation}\label{indepAtr}
\frac{p_{AB}}{p_{A}p_{B}} = 1, \,\,\frac{ p_{AC}}{p_{A}p_{C}} = 1, \,\, \frac{p_{BC}}{p_{B}p_{C}} = 1, \,\,\frac{p_{ABC}}{p_{A}p_{B}p_{C}} = 1.
\end{equation}
Here $p_A$, $p_B$, $p_C$, $p_{AB}$, $p_{AC}$, $p_{BC}$, $p_{ABC}$ denote positive probabilities of having the corresponding combination of features, and $p_A+p_B+p_C+p_{AB}+p_{AC}+p_{BC}+p_{ABC} = 1$. Let $\boldsymbol p = (p_A, p_B, p_C, p_{AB}, p_{AC}, p_{BC}, p_{ABC})'$. 
This is a relational model generated by the subsets 
$S_1$ (possessing the feature $A$), $S_2$ (possessing the feature $B$), and $S_3$ (possessing the feature $C$) and is a variant of independence when no unaffected cases exist. It is defined using non-homogeneous odds ratios in (\ref{indepAtr}) and thus does not have the overall effect. 
Let 
$\boldsymbol q=(0.04, 0.04, 0.04, 0.04, 0.04, 0.24, 0.56)'$ be the parameter of the observed probability distribution.  The vector sequence $\boldsymbol p^{(n)}$, produced by IIS, converges to $\boldsymbol p^*=(0.3202, 0.4574, 0.4574, 0.1464, 0.1464, 0.2092, 0.0670)'$, with the total $\boldsymbol 1 \boldsymbol p^*=1.804$. The normalized vector:
$$\boldsymbol p^*_1 =\boldsymbol p^*/(\boldsymbol 1 \boldsymbol p^*) = (
0.1775, 0.2535, 0.2535, 0.0812, 0.0812, 0.1160, 0.0371)'$$ does not have the multiplicative structure implied by (\ref{indepAtr}).

In order to illustrate that this problem is not isolated, rather quite common, the histogram of sums of the limit vectors obtained for 13352  observed distributions $\boldsymbol q$, generated with equal spacing on the parameter space, is given in Figure \ref{sumsIIS}.  

 \begin{center}   ***  Figure 1 here.  ***  \end{center}

In summary, a relational model with the overall effect can be parameterized using a model matrix whose rows do or do not sum to a constant vector. Depending on the parameterization, IPF, GIS or IIS can be used to compute the MLE under such a model. However, a relational model without the overall effect, under any parameterization, will have a model matrix whose rows do not sum to a constant vector, and, if the model is for probabilities, IPF, GIS, IIS do not converge to the MLE. 

The next section presents an iterative fitting procedure that converges to the MLE, whether or not the overall effect is present in the relational model.

\section{Iterative proportional fitting in relational models}\label{MLEipf}

Let $\mathbf{Y}$ be a random variable that has a Poisson distribution parameterized by $\boldsymbol \delta \equiv \boldsymbol \lambda$ or a multinomial distribution parameterized by $N$ and $\boldsymbol \delta \equiv \boldsymbol p$. Consider the relational model generated by a matrix $\mathbf{A}$, and  let $\mathbf{D}$ be a kernel basis matrix of it. Suppose $\boldsymbol y$ is a realization of $\mathbf{Y}$. Consider any  $\gamma > 0$, such that   $\mathcal{L}_{\boldsymbol \delta}(\mathbf{A}, \boldsymbol q, \gamma)$, with $\boldsymbol q$ given in (\ref{b}),  is not empty.

An iterative proportional fitting algorithm that computes the cell parameters, $\boldsymbol \delta_{\gamma}^*$, and the model parameters for the distribution that is the unique common point of the linear family $\mathcal{L}_{\boldsymbol \delta}(\mathbf{A}, \boldsymbol q, \gamma)$  and the multiplicative family $\mathcal{M}(\mathbf{A})$ is given next. 
\vspace{1mm}



\begin{center} \textbf{IPF($\gamma$) Algorithm:} \end{center}
\noindent {Set} $d = 0$; \,\, ${\delta}_{\gamma}^{(0)}(i) = 1$ for all $i \in \mathcal{I}$; \,\, $\theta_{\gamma}^{(0)}(j) = 1$ for all $j = 1, \dots, J$, and proceed as follows.
\begin{enumerate}
\item[] {\tt Step 1}: {Find} $j \in \{1,2,\dots,J\}$, such that $d+1 \equiv j \mbox{ mod } J$;
\item[] {\tt Step 2}: {Compute}  
\begin{eqnarray} 
\delta_{\gamma}^{(d+1)}(i) &=& \delta_{\gamma}^{(d)}(i) \left(\gamma \frac{A_{j}\boldsymbol{q}}{A_{j}\boldsymbol{\delta}_{\gamma}^{(d)}}\right)^{a_{ji}} \,\, \mbox{for all } i \in \mathcal{I};   \label{RipfGamma} \\
\theta_{\gamma}^{(d+1)}(j) &=& \theta_{\gamma}^{(d)}(j)\frac{\delta_{\gamma}^{(d+1)}({i}^*)}{\delta_{\gamma}^{(d)}({i}^*)}, \,\, \mbox{where } i^* \in \mathcal{I}: \,\,  a_{ji^*} >0.  \label{IPFpapamCompute}
\end{eqnarray}
\item[] {\tt Step 3}: While  $\gamma A_{j}\boldsymbol{q} \neq A_{j}\boldsymbol{\delta}_{\gamma}^{(d+1)}$ for at least one $j$,  set $d = d+1$, go to {\tt Step 1}.
\item[] {\tt Step 4}: Set $\boldsymbol{\delta}_{\gamma}^{*}=\boldsymbol{\delta}_{\gamma}^{(d)}$, and finish. \qed
\end{enumerate}

\vspace{3mm}

The proof of convergence of IPF($\gamma$) is based on showing that IPF($\gamma$) is an instantiation of the algorithm proposed by \cite{Bregman} to find a common point of convex sets.

Let $\boldsymbol t, \boldsymbol u \in \mathbb{R}^{|\mathcal{I}|}_{>0}\,$ and let $D(\boldsymbol t|| \boldsymbol u)$ denote the Bregman divergence associated with the function $F(\boldsymbol x) = \sum_{i \in \mathcal{I}} x(i) \mbox{log }x(i)$:
\begin{equation} \label{BDdef} 
D(\boldsymbol t|| \boldsymbol u) = \sum_{i \in \mathcal{I}} t(i) \mbox{log }(t(i)/u(i)) + (\sum_{i \in \mathcal{I}} u(i) - \sum_{i \in \mathcal{I}} t(i)).
\end{equation}
If $P$ and $Q$ are probability distributions parameterized by $\boldsymbol p$ and $\boldsymbol q$ respectively, then $D(\boldsymbol p|| \boldsymbol q) = I(P|| Q)$ is the Kullback-Leibler divergence between $P$ and $Q$.


\begin{lemma}\label{lemmaGamma}
For the sequence $\{\boldsymbol{\delta}_{\gamma}^{(d)}\}$ of vectors obtained with IPF($\gamma$),  \begin{equation}\label{Dproj1n}
D(\boldsymbol \delta^{(d+1)}|| \boldsymbol \delta^{(d)}) = \underset{ \boldsymbol \zeta: \, P_{\boldsymbol \zeta} \in \mathcal{L}_{\boldsymbol \delta} (A_j, \boldsymbol q , \gamma)}
{\mbox{min }} D(\boldsymbol \zeta|| \boldsymbol \delta^{(d)}).
\end{equation} 
\end{lemma}

\begin{proof}

The function $D(\boldsymbol \zeta|| \boldsymbol{\delta}_{\gamma}^{(d)})$ is convex with respect to $\boldsymbol \zeta$ and, therefore, its minimum, say $\boldsymbol \zeta^*$, on the set $\{\boldsymbol \zeta: \, P_{\boldsymbol \zeta} \in \mathcal{L}_{\boldsymbol \delta} (A_j, \boldsymbol q , \gamma)\}$ exists and is unique.  

Setting the derivatives of the Lagrangian 
\begin{eqnarray*}
{L} &=& D(\boldsymbol \zeta|| \boldsymbol{\delta}_{\gamma}^{(d)}) - \alpha({A}_{j} \boldsymbol \zeta - \gamma A_{j}\boldsymbol{q}) \\
&=& \sum_{i \in \mathcal{I}} \zeta(i) \mbox{log }(\zeta(i)/{\delta}_{\gamma}^{(d)}(i)) + (\sum_{i \in \mathcal{I}} {\delta}_{\gamma}^{(d)}(i) - \sum_{i \in \mathcal{I}} \zeta(i))- \alpha({A}_{j} \boldsymbol \zeta - \gamma A_{j}\boldsymbol{q})
\end{eqnarray*}
equal to zero, one obtains that $\boldsymbol \zeta^*$ is the unique solution to the equations
\begin{eqnarray}\label{EqDelta}
\mbox{log }(\zeta^*(i)/{\delta}_{\gamma}^{(d)}(i)) &=& \alpha a_{ji}, \,\, i \in \mathcal{I},\\ 
{A}_{j} \boldsymbol \zeta^* &=& \gamma A_{j}\boldsymbol{q}. \nonumber
\end{eqnarray}
Since $\boldsymbol{\delta}_{\gamma}^{(d+1)}$, given in (\ref{RipfGamma}), satisfies (\ref{EqDelta}) with $\alpha = \mbox{log }\gamma \frac{A_{j}\boldsymbol{q}}{A_{j}\boldsymbol{\delta}_{\gamma}^{(d)}}$, it is equal to $\boldsymbol \zeta^*$.
\end{proof}

Lemma \ref{lemmaGamma}  implies that, for each $d \geq 0$, the distribution parameterized by $\boldsymbol{\delta}_{\gamma}^{(d+1)}$ is the D-projection of the distribution parameterized by  $\boldsymbol{\delta}_{\gamma}^{(d)}$ on the set $\mathcal{L}_{\boldsymbol \delta} (A_j, \boldsymbol q , \gamma)$, that is, $\{\boldsymbol{\delta}_{\gamma}^{(d)}\}$ is a relaxation sequence with respect to Bregman divergence. Note that the distribution parameterized by $\boldsymbol \delta_{\gamma}^{(d+1)}$ does not necessarily have the same total as the distribution parameterized by $\boldsymbol \delta_{\gamma}^{(d)}$ did and, even if $\boldsymbol \delta \equiv \boldsymbol p$, is not necessarily a probability distribution.

It is proved in the following theorem that IPF($\gamma$) converges to the parameter $\boldsymbol{\delta}_{\gamma}^{*}$ of the unique common point of the linear family $\mathcal{L}_{\boldsymbol \delta}(\mathbf{A}, \boldsymbol q, \gamma)$ and the multiplicative family $\mathcal{M}(\mathbf{A})$.

\begin{theorem}\label{ThGamma}
The sequence  $\boldsymbol{\delta}_{\gamma}^{(d)}$, obtained from IPF($\gamma$), converges, as $d \to \infty$, and 
the limit  $\boldsymbol{\delta}_{\gamma}^{*}$ satisfies 
\begin{tabbing}
(i) \= \hspace{2mm} $\mathbf{A}\boldsymbol{\delta}_{\gamma}^* = \gamma \mathbf{A} \boldsymbol q$, \\
(ii)\> \hspace{3mm} $\mathbf{D} \mbox{log } \boldsymbol{\delta}_{\gamma}^* = \boldsymbol 0$.\\
\end{tabbing}
\end{theorem}

\begin{proof}

(i) Since $\{\boldsymbol{\delta}_{\gamma}^{(d)} \}$ is a relaxation sequence with respect to the function $D(\boldsymbol t|| \boldsymbol u)$, it converges to a $\boldsymbol{\delta}_{\gamma}^{*}$  \citep[Theorem 1]{Bregman} which belongs to the set 
$$
\cap_{j=1}^{J} \{ \boldsymbol{\delta}_{\gamma}:  A_j\boldsymbol{\delta}_{\gamma} = \gamma A_j \boldsymbol{q} \},$$
and thus $$\mathbf{A}\boldsymbol{\delta}_{\gamma}^* = \gamma\boldsymbol A \boldsymbol{q}.$$ 

\vspace{1mm}

(ii) IPF($\gamma$) multiplies the current value in each cell that belongs to a subset $S_j$ by the ratio of the desired subset sum to the actual one. This transformation leaves the values of the generalized odds ratios unchanged. The formal argument proceeds by induction. Since ${\delta}_{\gamma}^{(0)}(i) = 1$, for all $i \in \mathcal{I}$, $\mathbf{D} \mbox{log } \boldsymbol {\delta}_{\gamma}^{(0)} = \boldsymbol 0$, and the statement holds for $d = 0$. Assume that $\mathbf{D} \mbox{log } \boldsymbol{\delta}_{\gamma}^{(d)} = \boldsymbol 0$ for a positive integer $d$. Set $C_j = \frac{\gamma {A}_{j} \boldsymbol{\delta}}{{A}_{j} \boldsymbol{\delta}_{\gamma}^{(d)}}$. Then, 
\begin{eqnarray*}
\mathbf{D} \mbox{log } \boldsymbol{\delta}_{\gamma}^{(d+1)} &=& \mathbf{D} \mbox{log }\left[\begin{array}{c}  
{\delta}_{\gamma}^{(d)}(1)\cdot C_j^{a_{j1}}\\
{\delta}_{\gamma}^{(d)}(2)\cdot C_j^{a_{j2}}\\
\vdots \\
{\delta}_{\gamma}^{(d)}(|\mathcal{I}|)\cdot C_j^{a_{j|\mathcal{I}|}}
\end{array}\right] = \mathbf{D} \left[\begin{array}{c}  
\mbox{log }{\delta}_{\gamma}^{(d)}(1)+ {a_{j1}}\mbox{log } C_j\\
\mbox{log }{\delta}_{\gamma}^{(d)}(2)+ {a_{j2}}\mbox{log } C_j\\
\vdots \\
\mbox{log }{\delta}_{\gamma}^{(d)}(|\mathcal{I}|)+{a_{j|\mathcal{I}|}}\mbox{log } C_j
\end{array}\right] \\
& & \\
&=&  \mathbf{D} \mbox{log } \boldsymbol{\delta}_{\gamma}^{(d)} + \mbox{log } C_j \mathbf{D}A'_{j} = \boldsymbol 0,
\end{eqnarray*}
as $\mathbf{D}$ is a kernel basis matrix and thus $\mathbf{D}A'_{j}= \boldsymbol 0$.
Therefore, $\mathbf{D} \mbox{log } \boldsymbol{\delta}_{\gamma}^{(d)} = \boldsymbol 0$ for all $d = 0, 1, 2 \dots$.
Finally, by continuity of matrix multiplication and logarithm, $\mathbf{D} \mbox{log } \boldsymbol{\delta}_{\gamma}^{*} = \boldsymbol 0$. 
\end{proof}

The next statement specifies when IPF($\gamma$) converges to the maximum likelihood estimates under a relational model.

\begin{corollary}\label{corAllIPF}
The following statements hold:
\begin{enumerate}
\item \label{Cor1} Let $\mathbf{Y}$ be a random variable that has a multinomial distribution with parameters $N$ and $\boldsymbol p$. If the overall effect is present, then the sequence obtained from IPF($1$) converges to the maximum likelihood estimate $\hat{\boldsymbol p}$ of $\boldsymbol p$ under the model. If the overall effect is not present, then, if the adjustment factor $\gamma^*$ is known, the sequence obtained from IPF($\gamma^*$) converges to the maximum likelihood estimate $\hat{\boldsymbol p}$ of $\boldsymbol p$ under the model.
\item \label{Cor2} Let $\mathbf{Y}$ be a random variable that has a Poisson distribution with parameter $\boldsymbol \lambda$. Then, whether or not the overall effect is present, the sequence obtained from IPF($1$) converges to the maximum likelihood estimate $\hat{\boldsymbol \lambda}$ of $\boldsymbol \lambda$ under the model. \qed
\end{enumerate} 
\end{corollary}

The conventional IPF procedure computes maximum likelihood estimates for the cell parameters, and it is often pointed out that this algorithm is not suitable for estimating the model parameters 
\citep[cf.][p.14]{FeinbergRinaldo2012supplementary}. In fact, this is not the case, and it will be proven next that  the sequence $\boldsymbol \theta_{\gamma}^{(d)}$, obtained from IPF($\gamma$), converges to estimates of the model parameters. 

\begin{theorem} \label{ConvToParam}
For all $j = 1, \dots, J$, the sequence $\theta_{\gamma}^{(d)}(j)$, obtained in (\ref{IPFpapamCompute}), converges, as $d \to \infty$, to a  $\theta_{\gamma}^{*}(j)$ and 
$$ {\delta}_{\gamma}^{*}(i) = \prod_{j = 1}^J (\theta_{\gamma}^{*}(j))^{a_{ji}}, \,\, \mbox{ for all }  i \in \mathcal{I}. $$
\end{theorem}

\begin{proof}

By the choice of the initial values, ${\delta}_{\gamma}^{(0)}(i) = \prod_{j = 1}^J (\theta_{\gamma}^{(0)}(j))^{a_{ji}}$, for all $i \in \mathcal{I}$. The further argument is by induction.

Let $d \geq 0$ and  $d+1 \equiv j \mbox{ mod } J$. By the induction hypothesis, 
$${\delta}_{\gamma}^{(d)}(i) = \prod_{j = 1}^J (\theta_{\gamma}^{(d)}(j))^{a_{ji}}, \,\, \mbox{ for all }  i \in \mathcal{I}.$$

During the $(d+1)$-st iteration, only the parameters of the cells that belong to the subset $S_j$, and the model parameter corresponding to $S_j$ are updated; ${\delta}_{\gamma}^{(d+1)}(i) = {\delta}_{\gamma}^{(d)}(i)$  for  $ i \notin S_j$,  and  ${\theta}_{\gamma}^{(d+1)}(l) = {\theta}_{\gamma}^{(d)}(l)$ for $l \ne j$. Therefore, for $ i \notin S_j $
\begin{align*}
{\delta}_{\gamma}^{(d+1)}(i) &= \prod_{l = 1}^J (\theta_{\gamma}^{(d+1)}(l))^{a_{li}}, 
\end{align*}
because for those parameters that changed during the last step of iteration,  the exponent $a_{li}$ is zero. Further, for $i \in S_j$
\begin{align*}
{\delta}^{(d+1)}(i) &= {\delta}_{\gamma}^{(d)}(i) \frac {\theta_{\gamma}^{(d+1)}(j)}{\theta_{\gamma}^{(d)}(j)} =
\prod_{l = 1}^J (\theta_{\gamma}^{(d)}(l))^{a_{li}}\frac {\theta_{\gamma}^{(d+1)}(j)}{\theta_{\gamma}^{(d)}(j)} = 
 \prod_{l = 1}^J (\theta_{\gamma}^{(d+1)}(l))^{a_{li}}.\\
\end{align*}
By the principle of induction,  for all $d \geq 0$,
\begin{equation}\label{ModD}
{\delta}_{\gamma}^{(d+1)}(i) = \prod_{j = 1}^J (\theta_{\gamma}^{(d+1)}(j))^{a_{ji}}, \,\, \mbox{ for all }  i \in \mathcal{I}.
\end{equation}

During the iterations, the value of the parameter associated with the subset $S_j$ is updated at a step $j+kJ$, for some non-negative integer $k$, and then remains constant till it is updated again in the step $j+(k+1)J$. The ratio of two subsequent values of this parameter equals
\begin{align*}
 \frac{\theta_{\gamma}^{(j+(k+1)J)}(j)}{\theta_{\gamma}^{(j+kJ)}(j)}=\frac{\theta_{\gamma}^{(j + (k+1)J)}(j)}{\theta_{\gamma}^{(j-1 +  (k+1)J)}(j)} \cdot  
\frac{\theta_{\gamma}^{(j-1 +  (k+1)J)}(j)}{\theta_{\gamma}^{(j-2 +  (k+1)J)}(j)} \cdots 
\frac{\theta_{\gamma}^{(j -(J-1)+(k+1)J)}(j)}{\theta_{\gamma}^{(j+ kJ)}(j)}, 
\end{align*}
which, by repeated application of (\ref{IPFpapamCompute}), may be written as
\begin{align}  \label{subsequ}
 \frac{\delta_{\gamma}^{(j+(k+1)J)}(i^*)}{\delta_{\gamma}^{(j+kJ)}(i^*)}=\frac{\delta_{\gamma}^{(j + (k+1)J)}(i^*)}{\delta_{\gamma}^{(j-1 +  (k+1)J)}(i^*)} \cdot  
\frac{\delta_{\gamma}^{(j-1 +  (k+1)J)}(i^*)}{\delta_{\gamma}^{(j-2 +  (k+1)J)}(i^*)} \cdots 
\frac{\delta_{\gamma}^{(j -(J-1)+(k+1)J)}(i^*)}{\delta_{\gamma}^{(j+ kJ)}(i^*)}, 
\end{align}
where $i^*$ is specified prior to (\ref{IPFpapamCompute}).

By Theorem  \ref{ThGamma}, as $d \to \infty$,  ${\delta}_{\gamma}^{(d)}(i)$ converges to a ${\delta}_{\gamma}^*(i)$ for every $i \in \mathcal{I}$. Therefore, the right hand side of (\ref{subsequ}) converges to $1$ and thus 
$$
 \frac{\theta_{\gamma}^{(j+(k+1)J)}(j)}{\theta_{\gamma}^{(j+kJ)}(j)}  \to 1,
$$
when  $d \to \infty$, i.e., for every fixed $j$, when $k \to \infty$.

This, by applying an elementary calculus argument, implies that for every $j = 1, \dots, J$,   ${\theta}_{\gamma}^{(d)}(j)$, obtained in (\ref{IPFpapamCompute}), converges to a ${\theta}_{\gamma}^{*}(j)$. From (\ref{ModD}), by continuity, 
$$ {\delta}_{\gamma}^{*}(i) = \prod_{j = 1}^J (\theta_{\gamma}^{*}(j))^{a_{ji}}, \,\, \mbox{ for all }  i \in \mathcal{I}. $$
\end{proof}

The G-IPF algorithm described next can be used for computing the maximum likelihood estimates under relational models both for probabilities and for intensities, with or without the overall effect. In the case when the model is for probabilities and does not include the overall effect, G-IPF complements IPF($\gamma$) with a step that computes the adjustment factor.

\vspace{1mm}


\begin{center} \textbf{G-IPF Algorithm:} \end{center}


\begin{itemize}
\item[] If $\boldsymbol \delta \equiv \boldsymbol \lambda$,  compute $\hat{\boldsymbol{\lambda}}$ using IPF(1), and finish.
\item[] If $\boldsymbol \delta \equiv \boldsymbol p$, compute $\boldsymbol{p}^*$ using IPF($1$). \\
If $\boldsymbol 1\boldsymbol{p}^* = 1$, set $\hat{\boldsymbol p} = \boldsymbol p^*$, and finish. 
Otherwise, \\
{compute} $\gamma_L = (\boldsymbol 1\mathbf{A}\boldsymbol q)^{-1}$, $\gamma_R = \mbox{min } \{1/A_1\boldsymbol q, \dots, 1/A_J\boldsymbol q\}$, and proceed as follows:
\begin{itemize}
\item[] {\tt Step 1}: {Find} $\boldsymbol \delta_{(\gamma_L +\gamma_R)/2}^*$ \hspace{2mm} using IPF($\gamma$). \\
\item[] {\tt Step 2}: While   $\boldsymbol 1\boldsymbol \delta_{(\gamma_L +\gamma_R)/2}^* \ne 1$, 
\begin{itemize}
\item[] \qquad {if }  $\boldsymbol 1\boldsymbol{\delta}_{(\gamma_{L} + \gamma_{R})/2}^* < 1$,  {set } $\gamma_{L}  = \frac{\gamma_{L} + \gamma_{R}}{2}$,  
\item[] \qquad {else } {set } $\gamma_{R}  = \frac{\gamma_{L} + \gamma_{R}}{2}$;
\item[] \qquad {go to {\tt Step 1}}.\\
\end{itemize} 
\item[] {\tt Step 3}: Set  $\hat{\boldsymbol{p}}= \boldsymbol{\delta}_{(\gamma_L +\gamma_R)/2}^*$, and finish. \qed
\end{itemize}
\end{itemize}

\vspace{3mm}

The next lemma justifies the initial choice of $\gamma_L$ and $\gamma_R$. 

\vspace{2mm}

\begin{lemma}\label{monoton}
Let $\gamma > 0$ and let $\boldsymbol{\delta}_{\gamma}$ be a solution to the system
of equations 
\begin{equation}\label{system}
\mathbf{A}\boldsymbol{\delta} = \gamma \mathbf{A} \boldsymbol q, \,\, \mathbf{D} \mbox{log } \boldsymbol{\delta} = 0. 
\end{equation}
Then, for $\gamma_L = (\boldsymbol 1 \mathbf{A} \boldsymbol q)^{-1}$ and $\gamma_R =\mbox{min } \{1/{A}_1 \boldsymbol q, \dots, 1/{A}_J \boldsymbol q\}$,
$\boldsymbol 1\boldsymbol \delta_{\gamma_L} \leq 1$ and $\boldsymbol 1\boldsymbol \delta_{\gamma_R} \geq 1$.
\end{lemma}

\begin{proof}
Each column of $\mathbf{A}$ contains at least one $1$, and thus
$$\boldsymbol 1\boldsymbol \delta_{\gamma} \, \leq  \boldsymbol 1\mathbf{A}\boldsymbol \delta_{\gamma}= \gamma \boldsymbol 1\mathbf{A}\boldsymbol q,$$
which implies the first inequality.  Since $\mathbf{A}$ is a 0-1 matrix, 
$$\gamma \mbox{max } \{{A}_1 \boldsymbol q, \dots, {A}_J \boldsymbol q\} = \mbox{max } \{{A}_1 \boldsymbol \delta_{\gamma}, \dots, {A}_J \boldsymbol \delta_{\gamma}\} \, \leq \boldsymbol 1\boldsymbol \delta_{\gamma}. $$
Using that   
$$(\mbox{max } \{{A}_1 \boldsymbol q, \dots, {A}_J \boldsymbol q\})^{-1}=\mbox{min } \{1/{A}_1 \boldsymbol q, \dots, 1/{A}_J \boldsymbol q\},$$ 
implies the second inequality.
\end{proof}

The next lemma is needed for the proof of convergence of G-IPF.

\begin{lemma} \label{Continuity}
The solution, $\boldsymbol \delta_{\gamma}$, to the system of equations (\ref{system}) 
and thus its total, $\boldsymbol 1 \boldsymbol \delta_{\gamma}$, are continuous functions of $\gamma$. 
\end{lemma}

\begin{proof}
Fix an arbitrary $\gamma_0 > 0$. It will be shown that for a sequence $\{\gamma_n\} > 0$, which converges to $\gamma_0$, as $n \to \infty$,  $\boldsymbol \delta_{\gamma_n} \to \boldsymbol \delta_{\gamma_0}$.

Since  $\gamma_n$ converges, there exists a finite cover of $\{\gamma_n\}$ and so  does a finite cover of the sequence of segments $\{[0,\gamma_n]\}$. Since $\boldsymbol 1 \boldsymbol q = 1$, $\boldsymbol \delta_{\gamma_n} \leq \gamma_n$, for every $n \geq 1$. Hence, there exists a finite cover of the sequence $\{\boldsymbol \delta_{\gamma_n}\}$, implying that the set $\{\boldsymbol \delta_{\gamma_n}\}$ is compact, and, therefore,  there exists a subsequence $\{\boldsymbol \delta_{\gamma_{n_k}}\}$ that converges to a $\boldsymbol \delta^*$. By continuity of the logarithm and matrix multiplication,
$$\mathbf{A} \boldsymbol \delta^* = \gamma_0 \mathbf{A} \boldsymbol q, \quad \mathbf{D} \mbox{log } \boldsymbol \delta^* = \boldsymbol 0.$$
Since the solution of this system is unique, $\boldsymbol \delta^* = \boldsymbol \delta_{\gamma_0}$, and thus $\boldsymbol \delta_{\gamma_{n_k}} \to \boldsymbol \delta_{\gamma_0}$. \end{proof}

The next theorem states that G-IPF produces maximum likelihood estimates of the cell parameters under a relational model, whether or not the overall effect is present. 

\begin{theorem}\label{newGIPFconvTh}
Assume that the maximum likelihood estimate of $\boldsymbol \delta$ under the relational model generated by $\mathbf{A}$ exists and is unique. Then $\hat{\boldsymbol{\delta}}$, obtained from G-IPF, is the maximum likelihood estimate of $\boldsymbol \delta$.
\end{theorem}

\begin{proof}

In the case when $\boldsymbol \delta \equiv \boldsymbol \lambda$, G-IPF computes  $\hat{\boldsymbol{\lambda}}$ and $\hat{\boldsymbol{\theta}}$ using IPF(1). By Corollary \ref{corAllIPF},  $\hat{\boldsymbol{\lambda}}$ is the MLE of the cell intensities.
 
In the case when $\boldsymbol \delta \equiv \boldsymbol p$, G-IPF first calls IPF(1) and computes a vector of probabilities $\boldsymbol{p}^*$. 

If $\boldsymbol 1 \boldsymbol{p}^* = 1$, which happens if and only if $\boldsymbol 1 \in R(\mathbf{A})$, then, by Corollary \ref{corAllIPF}, $\hat{\boldsymbol{p}} = \boldsymbol{p}^*$ is the MLE of the cell probabilities.

Suppose $\boldsymbol 1 \boldsymbol{p}^* \ne 1$. By Lemma \ref{monoton}, the function $\boldsymbol 1 \boldsymbol{p}_{\gamma} - 1$ has opposite signs at the endpoints of the segment $[(\boldsymbol 1\mathbf{A}\boldsymbol q)^{-1}, \,\, {\mbox{min }} \{1/{A}_1 \boldsymbol q, \dots, 1/{A}_J \boldsymbol q \}]$. Since $\boldsymbol 1 \boldsymbol{p}_{\gamma} - 1$ is a continuous function of $\gamma$, the segment contains a
value 
\begin{equation}\label{gammaStar}
\gamma^*, \mbox{ such that } \boldsymbol 1 \boldsymbol{p}_{\gamma^*} = 1.
\end{equation}
In order to find $\gamma^*$ the algorithm uses the bisection method \citep[cf.][]{Boyd}. By the uniqueness of the MLE, $\hat{\boldsymbol p} = \boldsymbol{p}_{\gamma^*}$. 
\end{proof}

\vspace{3mm}

A reviewer to an earlier version of the paper pointed out that the above algorithm cannot be programmed because any implementation of IPF($\gamma$) will return the value of $\boldsymbol \delta_{\gamma}^*$ with some error. A programmable variant of the algorithm is obtained by choosing a desired precision $\epsilon > 0$.  Define IPF($\gamma, \epsilon$) by replacing {\tt Steps 3} and {\tt 4} of IPF($\gamma$) by the following steps:
\begin{itemize}
\item[] {\tt Step 3}: While  $|\gamma A_{j}\boldsymbol{q} - A_{j}\boldsymbol{\delta}_{\gamma}^{(d+1)}| > \epsilon$ for at least one $j$,  set $d = d+1$, go to {\tt Step 1}.
\item[] {\tt Step 4}: Set $\boldsymbol{\delta}_{\gamma, \epsilon}^{*}=\boldsymbol{\delta}_{\gamma}^{(d)}$, and finish. 
\end{itemize}
The sequence of vectors $\boldsymbol{\delta}_{\gamma, \epsilon}^{*}$, obtained using IPF($\gamma$, $\epsilon$), converges pointwise, as $\epsilon \to 0$, to  $\boldsymbol{\delta}_{\gamma}^*$, the result of IPF($\gamma$).

\vspace{5mm}

The G-IPF($\epsilon$) algorithm described next computes the maximum likelihood estimates of cell parameters under a relational model with a given precision. 

\vspace{1mm}


\begin{center} \textbf{G-IPF($\epsilon$) Algorithm:} \end{center}

\begin{itemize}
\item[] If $\boldsymbol \delta \equiv \boldsymbol \lambda$,  compute $\hat{\boldsymbol{\lambda}}_{\epsilon}$ using IPF(1, $\epsilon$), and finish.
\item[] If $\boldsymbol \delta \equiv \boldsymbol p$, compute $\boldsymbol{p}^*_{\epsilon}$ using IPF($1$, $\epsilon$). \\
If $|\boldsymbol 1\boldsymbol{p}^*_{\epsilon} - 1| < \epsilon$, set $\hat{\boldsymbol p}_{\epsilon} = \boldsymbol p^*_{\epsilon}$, and finish. 
Otherwise, set $T = 1$, \\
{compute} $\gamma_{L} = (\boldsymbol 1\mathbf{A}\boldsymbol{q})^{-1}$, $\gamma_{R} = \mbox{min } \{1/A_1\boldsymbol{q}, \dots, 1/A_J \boldsymbol{q}\}$, and proceed as follows: 
\item[] {\tt Step 1}: for $t = 0, \dots, T$, let $\gamma_t = \gamma_L + \frac{t}{T}(\gamma_R - \gamma_L)$,  and compute $\boldsymbol p_{\gamma_t, \epsilon}^*$ using IPF($\gamma_t$, $\epsilon/T$). Whenever $|\boldsymbol 1 \boldsymbol p_{\gamma_t, \epsilon}^* - 1| < \epsilon$, go to {\tt Step 3}.
\item[] {\tt Step 2}: Set $T = T+1$, go to {\tt Step 1}.
\item[] {\tt Step 3}: Set  $\hat{\boldsymbol{p}}_{\epsilon}= \boldsymbol p_{{\gamma}_{t},\epsilon}^*$, $\hat{\gamma}_{\epsilon} = \gamma_t$, and finish. \qed
\end{itemize}

The next theorem states that G-IPF($\epsilon$) produces maximum likelihood estimates of the cell parameters under a relational model with precision $\epsilon$.

\begin{theorem}\label{G_IPFepsConv}
Assume that the maximum likelihood estimate of $\boldsymbol \delta$ under the relational model generated by $\mathbf{A}$ exists and is unique. 
For every $\epsilon > 0$, G-IPF($\epsilon$) terminates and returns $\hat{\boldsymbol{\delta}}_{\epsilon}$ and $\hat{\gamma}_{\epsilon}$, which satisfy:
\begin{align}
&|\mathbf{A} \hat{\boldsymbol{\delta}}_{\epsilon}  - \hat{\gamma}_{\epsilon}\mathbf{A} {\boldsymbol{q}}| < \epsilon; \label{EpsEq} \\
&\mathbf{D} \mbox{log } \hat{\boldsymbol{\delta}}_{\epsilon}= 0; \label{EpsDlog} \\
&|\boldsymbol 1 \hat{\boldsymbol{\delta}}_{\epsilon} - 1| < \epsilon \quad (\mbox{for } \boldsymbol \delta \equiv \boldsymbol p).  \label{normalizEps}
\end{align}
\end{theorem}

\begin{proof}

Fix an $\epsilon > 0$. When $\boldsymbol \delta \equiv \boldsymbol \lambda$, G-IPF($\epsilon$) computes  $\hat{\boldsymbol{\lambda}}_{\epsilon}$ using IPF($1$, $\epsilon$). Since IPF($1$, $\epsilon$) terminates, so does G-IPF($\epsilon$). The termination condition of IPF($1$, $\epsilon$) implies (\ref{EpsEq}), with $\hat{\gamma}_{\epsilon} = 1$. 

In the case when $\boldsymbol \delta \equiv \boldsymbol p$, G-IPF($\epsilon$) first calls IPF($1$, $\epsilon$), which terminates and results in $\boldsymbol{p}^*_{\epsilon}$, for which (\ref{EpsEq}) holds with $\hat{\gamma}_{\epsilon} = 1$. If (\ref{normalizEps}) also holds, G-IPF($\epsilon$) is terminated after the first iteration with $\hat{\boldsymbol p}_{\epsilon} = \boldsymbol{p}^*_{\epsilon}$. 

Suppose $|\boldsymbol 1 \boldsymbol{p}^*_{\epsilon} - 1| > \epsilon$, and thus G-IPF($\epsilon$) does not stop after the first iteration. By  Lemma \ref{monoton}, the (unknown) adjustment factor corresponding to the MLE, say $\hat{\gamma}$, belongs to the segment  $[\gamma_L, \gamma_R]$. By Lemma \ref{Continuity},  $\,\boldsymbol 1 \boldsymbol p_{\gamma} - 1$ is a continuous function of $\gamma$, and, since $\boldsymbol 1\boldsymbol p_{\hat{\gamma}} = 1$, for every $T_0 \in \mathbb{Z}_{>0}$, there exists a $t_0 \in \{0 ,\dots, T_0\}$, such that $\hat{\gamma} \in [\gamma_{t_0}, \gamma_{t_0+1}]$, where $\gamma_{t_0} = \gamma_L + \frac{t_0}{T_0}(\gamma_R - \gamma_L)$. As $T_0 \to \infty$, the length of the interval $[\gamma_{t_0}, \gamma_{t_0+1}]$ approaches $0$, and thus, for a large enough $T_0$, both, $\boldsymbol p_{t_0, \epsilon}^*$ and $\boldsymbol p_{t_0+1, \epsilon}^*$, obtained from IPF($\gamma_{t_0}$, $\epsilon/T_0$)  and  IPF($\gamma_{t_0+1}$, $\epsilon/T_0$), respectively, satisfy (\ref{normalizEps}). Thus, G-IPF($\epsilon$) will terminate at the $(T_0+1)$-st iteration with $\hat{\boldsymbol p}_{\epsilon}$ that fulfills (\ref{EpsEq}), with some $\hat{\gamma}_{\epsilon}$, and  (\ref{normalizEps}).

Since the multiplicative structure of the distributions does not change during the iterations, (\ref{EpsDlog}) holds. 
\end{proof}

The theorem implies that the sequence of vectors $\hat{\boldsymbol{\delta}}_{\epsilon}$, obtained using G-IPF($\epsilon$), converges pointwise, as $\epsilon \to 0$, to the maximum likelihood estimate of $\boldsymbol \delta$. 

The G-IPF($\epsilon$) algorithm is implemented in the R-package {\tt{gIPFrm}} \citep{gIPFpackage}. In this implementation, there is also an option to update the adjustment factor using the bisection method.  Alternatively, any numerical technique available for monotone functions \citep[cf.][]{Boyd} may be used.



\nocite{Mccallum2000}
\nocite{Yin}

\nocite{Yin}
\nocite{ClassenDrivErr}

\section*{Acknowledgments}

Part of the material presented here was contained in the PhD thesis of the first author to which the second author and Thomas Richardson were advisers.
The authors wish to thank him for several comments and suggestions. The proof of Proposition \ref{featuresVSoverEff} uses the idea of Olga Klimova, to whom the authors are also indebted. The second author was supported in part by Grant K-106154 from the Hungarian National Scientific Research Fund (OTKA).

\bibliographystyle{plainnat}

\bibliography{uwthesis0507}

\begin{figure}
\label{sumsIIS}
\begin{center}
\includegraphics[scale=0.45]{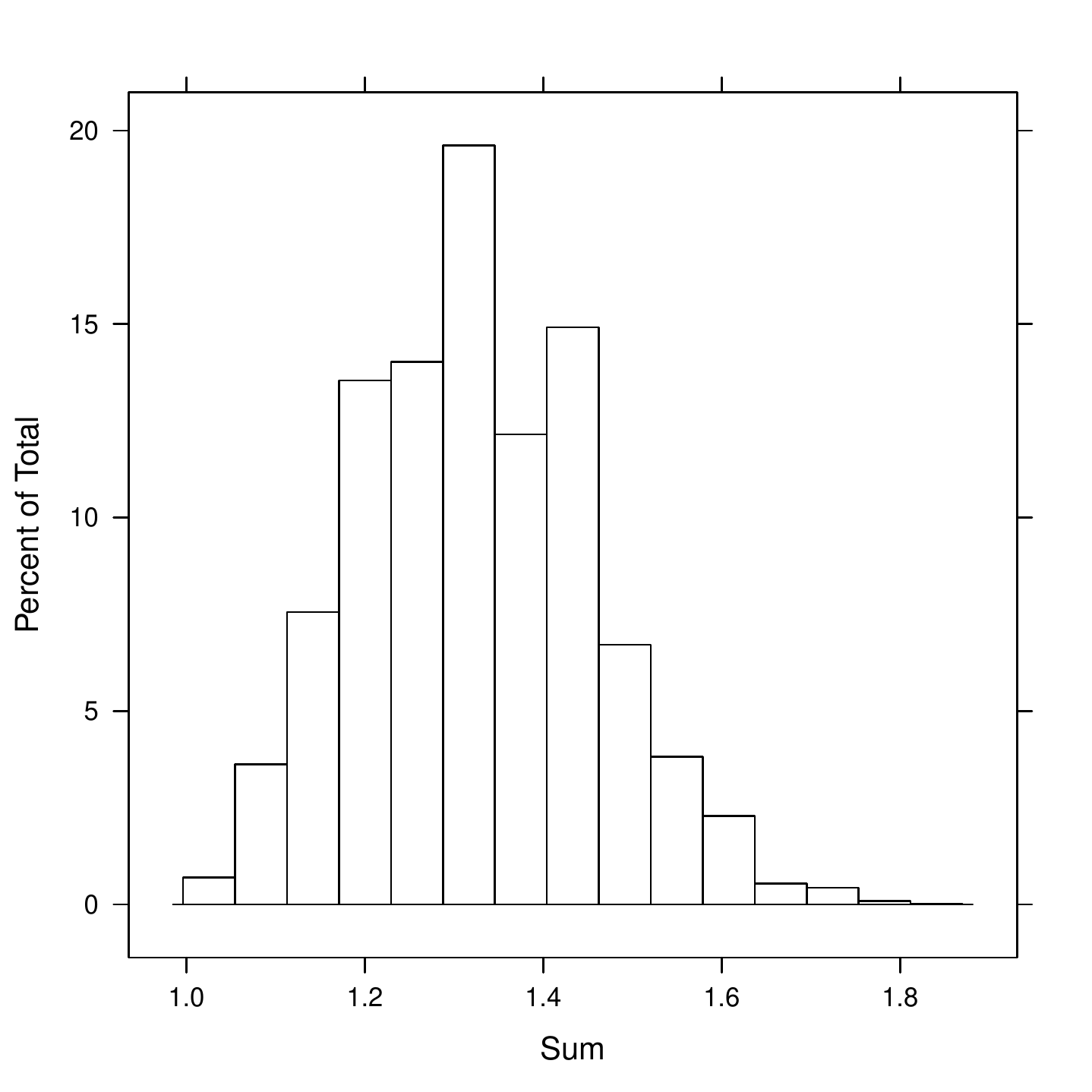}
\end{center}
\caption{Histogram of sums of the limit vectors obtained by IIS, under the model of \cite{AitchSilvey60}, for 13352 generated distributions.}
\end{figure}

\vspace{10mm}

\begin{table}[h]
\centering
\caption{Relational models and properties of the MLE  \citep{KRD11}.}
\label{ModelTypes}
\vspace{5mm}
\setlength{\extrarowheight}{5pt}
{
\begin{tabular}{|m{36mm}|m{34mm}|m{32mm}|m{32mm}|}
\hline
&
\multicolumn{1}{m{34mm}|}{{Models with the overall effect}} &
\multicolumn{2}{m{64mm}|}{ {Models without the overall effect}}\\
\cline{2-4}

 & \multicolumn{1}{c|}{Probabilities $\&$ Intensities} & \multicolumn{1}{c|}{ Probabilities} & \multicolumn{1}{c|}{ Intensities}\\
\hline
{ Exponential family } & \multicolumn{1}{c|}{Regular}	& \multicolumn{1}{c|}{{Curved}} 	 & \multicolumn{1}{c|}{Regular}\\
\hline
{ Subset sums of the MLE vs observed subset sums }&  \multicolumn{1}{c|}{Equal} 	& \multicolumn{1}{c|}{{Proportional}} 	& \multicolumn{1}{c|}{Equal}\\
\hline
{ Adjustment for total } & \multicolumn{1}{c|}{1}	& \multicolumn{1}{c|} {1}& {{Depends on the data}}\\
\hline
{ Adjustment for subset sums } & \multicolumn{1}{c|}{1}	& {{Depends on the data}} 	 & \multicolumn{1}{c|}{1}\\
\hline
\end{tabular}}
\end{table}

\end{document}